\documentclass[a4paper,11pt,reqno]{amsart}

\usepackage{hyperref}
\usepackage{amsmath,amssymb}
\usepackage{graphicx}
\usepackage{multicol}
\usepackage{enumerate}
\usepackage[T1]{fontenc}

\usepackage{bbold}

\newtheorem{lemma}{Lemma}[section]
\newtheorem{theorem}[lemma]{Theorem}
\newtheorem{proposition}[lemma]{Proposition}

\theoremstyle{definition}
\newtheorem{remark}[lemma]{Remark}

\newtheorem{assumption}[lemma]{Assumption}

\numberwithin{equation}{section}

\newcommand{\al}{\alpha}
\newcommand{\de}{\delta}

\renewcommand{\Im}{\operatorname{Im}}
\renewcommand{\Re}{\operatorname{Re}}
\newcommand{\cA}{\mathcal{A}}

\DeclareMathOperator{\RE}{{\mathcal R}e}

\def\be{\begin{equation}}
\def\ee{\end{equation}}
\def\qquad{{\quad\quad}}

\def\RE{\mathbb R}
\def\CO{{\mathbb C}}

\def\sgn{\text{\rm sgn}}
\def\S{\mathcal S}
\def\q{\underline q}
\def\z{\underline z}
\def\II{\mathbb{1}}

\title[The 1-D Dirac equation with concentrated nonlinearity]{The 1-D Dirac equation with concentrated nonlinearity}

\author{Claudio Cacciapuoti}
\address{DiSAT, Sezione di Matematica, Universit\`a dell'Insubria, via Valleggio 11, I-22100
Como, Italy}
\email{claudio.cacciapuoti@uninsubria.it}
\author{Raffaele Carlone}
\address{Dipartimento di Matematica e Applicazioni R. Caccioppoli,  Universit\`a Federico II di Napoli, MSA, via Cinthia, I-80126, Napoli, Italy}
\email{raffaele.carlone@unina.it}
\author{Diego Noja}
\address{Dipartimento di Matematica e Applicazioni, Universit\`a di Milano Bicocca, via R. Cozzi 53, I-20125 Milano, Italy}
\email{diego.noja@unimib.it}
\author{Andrea Posilicano}
\address{DiSAT, Sezione di Matematica, Universit\`a dell'Insubria, via Valleggio 11, I-22100
Como, Italy}
\email{andrea.posilicano@uninsubria.it}

\thanks{
C.C. and R.C.  acknowledge the support of the FIR 2013 project ``Condensed Matter in Mathematical Physics'', Ministry of University and
Research of Italian Republic  (code RBFR13WAET)}

\begin{document}

\begin{abstract}
We define and study the Cauchy problem for a 1-D nonlinear Dirac equation with nonlinearities concentrated at one point. Global well-posedness is provided and conservation laws for mass and energy are shown. Several examples, including nonlinear Gesztesy-\v{S}eba  models and the concentrated versions of the Bragg Resonance, Gross-Neveu, and Soler type models, all within the scope of the present paper, are given. The key point of the proof consists in the reduction of the original equation to a nonlinear integral equation for an auxiliary, space-independent variable (the ``charge'').
\end{abstract}


\maketitle

\begin{footnotesize}
 
 \end{footnotesize}

\vspace{1cm}

\section{Introduction.}

\noindent
The interest in the nonlinear Dirac equation traces back to the paper \cite{Thirring:1958gw}, where a soluble nonlinear quantum field model in 1+1 space-time dimensions for self-interacting fermions was introduced. Other well known quantum field theoretic examples are given in \cite{Soler:1970}, again describing a self-interacting electron in 3+1 spacetime dimensions, and later in \cite{Gross:1974gs}, this one describing a model related to quantum chromodynamics. However, the nonlinear Dirac equation appears  also as an effective equation in condensed matter physics, here describing localization effects for solutions of Nonlinear Schr\"odinger or Gross-Pitaevskii equation in small periodic potentials (see e.g. \cite{goodman:2001wh} and the monography \cite{Pel:2011} for extended description and bibliography). Relevant applications are in photonic crystals and in Bose Einstein condensates, where a 2-D nonlinear Dirac equation plays the role of an effective equation governing the evolution of wavepackets spectrally concentrated near Dirac points of a periodic optical lattice (see \cite{AZ:12, DeSterke:2008ty,  FW:12, goodman:2001wh}  and references therein). Inspired by the above models, the rigorous analysis of the Dirac equation with general nonlinearities is now a major subject. As regards regularity well-posedness results we only mention some of the relevant papers as \cite{Bournaveas:2008wl, Candy:2011vk, Delgado:1978be, goodman:2001wh, Machihara:2010eo, Selberg:2010uv}. Specifically, for the 1-D case, results about the global well-posedness in the Sobolev space $H^{1}(\mathbb{R})$ for several types of nonlinearities  are known. For a review about the global well-posedness of nonlinear Dirac equation in one space dimension see \cite{Pelinovsky:2010vh}.\par
In this paper we define and solve the Cauchy problem for a Dirac-type equation with concentrated nonlinearity. By this we mean that the nonlinearity is space dependent and acts at a single point in space. Models of this type are popular in physics in the case of the Schr\"odinger equation (see e.g. \cite{Bulashenko:1996cd, Dror:2011kr, Malomed:1993bg}), and there is also a growing literature of mathematical character related to their well-posedness \cite{Adami:2003te, Adami:2004te, Adami:2001bt} see also \cite{HolmerLiu}, orbital and asymptotic stability of standing waves see \cite{Cecilia1}, and approximation through smooth space dependent nonlinearities \cite{Claudio, Cacciapuoti:2014gt, CFNTrxv15} see also \cite{Komech:2007ia}. A related work on the wave equation in dimension three is \cite{noja2005wave}. To the knowledge of present authors, there is no analogue activity related to the Dirac equation and this paper is possibly a first contribution to the subject (see however the interesting paper \cite{KK:10} where a model which represents a regularization of a concentrated nonlinearity is considered). To introduce the problem in the simplest way, (details are given in the following sections) we consider the Dirac operator 
\[
D_{m}\Psi:=-i\hbar\,c\,\sigma_{1}\frac{d\Psi}{dx}+m\, c^{2}\sigma_{3}\Psi
\]
where $\sigma_{1}$ and $\sigma_3$ is a suitable choice of Pauli matrices.
The nonlinear Dirac equation with a space dependent nonlinearity is given by 
\[
i\hbar\,\frac{\partial}{\partial t}\Psi=D_{m}\Psi +V\,g(\Psi)\,,
\]
where $V = V(x)$. In this paper we would  ideally treat the case where $V\to \delta_y$ weakly. This limit procedure can be consistently pursued in the case of  nonlinear Schr\"odinger equation, and yields to a well defined, non trivial and nonlinear dynamics (see \cite{Claudio, Cacciapuoti:2014gt} and references therein). The corresponding three dimensional model has been also studied mainly from the mathematical point of view (see \cite{Adami:2003te, Adami:2004te, Cecilia1,CFNTrxv15}).
The same constructive analysis could be attempted for the Dirac equation, but here we make use of a more abstract approach which has the virtue of complete generality. The starting point is the construction of linear singular perturbations of the Dirac operator, well known for a long time (see \cite{SAlbeverio:2011vta, Benvegnu:1994ci, Carlone:2013ht,  Gesztesy:1987gn}). 
The idea is to restrict the free Dirac operator $D_{m}$ to regular functions out of the point $y$, obtaining a symmetric, not self-adjoint operator. The self-adjoint extensions of this operator give rise to a unitary dynamics. Among them there is of course the Dirac operator itself, but many others exist which differ for the singular behavior at the point $y$. They are parametrized through a singular boundary condition embodied in the domain of the extended operator:
$$
\{\Psi\in H^{1}({\mathbb R}\backslash\{y\})\otimes\CO^{2}: \ \ ic\,\sigma_{1}[\Psi]_{y} =A\q\}\,,
$$
where the  two component vector $\q:=\frac12(\Psi(y^{+})+\Psi(y^{-}))$, also called  {\it charge}, is the mean value,  and $[\Psi]_{y}:=\Psi(y^{+})-\Psi(y^{-})$ is the jump of the spinor $\Psi$ at $y$, while $A$ is any $2\times 2 $ Hermitian matrix. The case $A=0$ gives of course the free Dirac operator on the line, while in all other cases there are singularities at the point $y$, because the jump of $\Psi$ is non trivial. It is easy to see that one ends up with the evolution described by the distributional equation
$$
 i\hbar\frac{d}{d t}\Psi(t) = D_{m}\Psi+i\hbar\, c\, \sigma_{1}[\Psi]_{y}\delta_{y}
$$
with $\Psi$ belonging to the above domain.\par
To define a nonlinear dynamics, we let the matrix $A$ dependent on the charge $\q$, arriving to a nonlinear operator $H^{nl}_{A}$ with a domain characterized by a nonlinear boundary condition at the point $y$:
\[
D(H^{nl}_{A})  = 
\left\{\Psi \in \mathcal{H} :\Psi \in H^{1}(\RE\backslash\{y\}) \otimes \CO^{2},\,ic \,\sigma_{1}[\Psi ]_{y} = A(\q) \q  \right\},
\]
where the matrix $A(\underline z)$ is Hermitian for all $\underline z\in \CO^2$. \par
Under a technical condition, see Assumption \ref{a:F}, we will show the following well-posedness result (see Th. \ref{t:main} for the precise statement):\\ 

For any $\Psi_\circ\in D(H^{nl}_{A})$ there exists a unique, global in time, solution $\Psi(t)$ of the Cauchy problem
\[
\begin{cases}
i\hbar\frac{d}{d t}\Psi(t)=H^{nl}_{A}\Psi(t)=D_{m}\Psi+\hbar A(\q)\q\,\delta_{y}
&\\
\Psi(0)=\Psi_{\circ}\in D(H^{nl}_{A}) \,.
\end{cases}
\]
\medskip 

%
%
%
A relevant fact about the proof of the main theorem is that, re-casting the initial value problem in integral form through the Duhamel formula,  $\Psi(t)$  turns out to depend on the solution of a nonlinear integral equation (giving the evolution of the charge $\q$, see equation \eqref{qnl}), which rules the behavior of the system. Once the solution of this nonlinear integral equation is guaranteed, a representation formula for the solution of the Cauchy problem (which seems to be new even in the linear case) allows to close the proof of the theorem. The Assumption \ref{a:F} on $A(\z)$ is needed to treat existence and uniqueness of the solution of the charge equation, and the rest of the proof consists in assuring the stated regularity properties of the solution. \par To conclude the introduction we now give a brief outline of the content of the various sections of the paper.\par
In the preliminary Section 2, in order to render self-contained the presentation, we recall the definition of the 1-D Dirac equation with a linear point interaction. Here we also provide a new representation formula for the solution of the linear Cauchy problem (see Proposition \ref{p:cauchylin}).\par
Section 3 is the core of the paper. The definition of the Dirac operator perturbed by a concentrated nonlinearity is given and it is shown how to split the nonlinear flow in the sum of the free flow plus a part containing the charge only (depending on the total initial datum) which satisfies  a nonlinear integral equation. It is then shown that the charge equation in the stated hypotheses admits unique solution and the main theorem is proved. The Section ends with the proof of three complementary but relevant properties. It is proven the independence of global well-posedness results on the special representation of the algebra of Dirac matrices employed, see Rem. \ref{equiv}; the mass (or $L^2$-norm), see Th. \ref{t:mass}, and energy  conservation laws, see Th. \ref{t:energy}, are shown. As regards the energy, in order to get conservation one has to restrict the admissible matrix fields $A(\z)$ imposing the constraint $A(\z)={\mathcal A}(\bar\z,\z)={\mathcal A}(\z,\bar\z)$.\par
In Section 4 several examples are given, also but not only in the effort to make contact with usual nonlinearities treated in the literature. Among others, the nonlinear versions of the Gesztesy-\v Seba models and the concentrated nonlinearities mimicking the Soler-type, Gross-Neveu and Bragg resonance models are treated. Finally, in Appendix A the representation formula for the free Dirac evolution in one dimension is recalled and 
the $H^{1}$-regularity in time of the evaluation at the singularity of the free part of the evolution is proved, see Prop. \ref{H1}.\\

Throughout the paper we shall use the following notation 
\begin{itemize}
\item
The inner  product between two vector valued functions in $L^2(\mathbb{R})\otimes\mathbb{C}^2$ is denoted by $\langle\Psi,\Phi\rangle$, and it is antilinear in the first argument. The corresponding norm is simply denoted by $\|\Psi\|$.
\item
The inner  product between two vectors in $\mathbb{C}^2$ is denoted by $\langle \z,\underline{\xi}\rangle_{\CO^2}$, and it is antilinear in the first argument. The corresponding norm is simply denoted by $|\z|$, not to be mistaken with the usual absolute value which is denoted in the same way.
\item
$C$ denotes a generic positive constant whose value may change from line to line. 
\end{itemize}

\section{The Cauchy problem for the Dirac equation with point interactions.\label{s:lin}}
Let $D_{m}:\S'(\RE)\otimes\CO^{2}\to \S'(\RE)\otimes\CO^{2}$ be the differential operator  
\[
D_{m}\Psi:=-i\hbar\,c\,\sigma_{1}\frac{d\Psi}{dx}+m\, c^{2}\sigma_{3}\Psi
\]
corresponding to the free 1-D Dirac operator with mass $m\ge 0$. Here $\hbar$ is Plank's constant, $c$ the light velocity, $\S'(\RE)$ denotes the space of tempered distribution, $\Psi=\begin {pmatrix}\psi_1 \\\psi_2\end {pmatrix}$ and $\sigma_{1}$ and $\sigma_{3}$ are the  first and the third among the  three Pauli matrices
\[
\sigma_{1}=\begin {bmatrix}0& 1 \\1 & 0\end {bmatrix},\qquad 
\sigma_{2}=\begin {bmatrix}0& -i \\i & 0\end {bmatrix},\qquad
\sigma_{3}=\begin {bmatrix}1& 0 \\0 & -1\end {bmatrix}.
\]
On the Hilbert space $L^2(\mathbb{R})\otimes\mathbb{C}^2$, the linear operator 
$$
H: D(H)\subset L^2(\mathbb{R})\otimes\mathbb{C}^2\to L^2(\mathbb{R})\otimes\mathbb{C}^2
\,,\qquad H\Psi:=D_{m}\Psi$$ 
with domain $D(H)=H^{1}(\RE)\otimes\CO^{2}$ is self-adjoint, where $H^{1}(\RE)$ denotes the Sobolev space of square integrable functions with square integrable  first order distributional derivatives.\par
Now we recall the construction of the self-adjoint singular perturbations of $H$ formally corresponding to the addition of a $\delta$-type potential (see e.g. \cite{SAlbeverio:2011vta, Benvegnu:1994ci, Carlone:2013ht, Gesztesy:1987gn}). \par
Given $y\in\RE,$ let $H_{-}$ and $H_{+}$ be the free Dirac  operators on $L^{2}(-\infty,y)\otimes\CO^2$ and $L^{2}(y,+\infty)\otimes\CO^2$ with domains 
$D(H_{-})=H^{1}(-\infty,y)\otimes\CO^{2}$ and $D(H_{+})=H^{1}(y,+\infty)\otimes\CO^{2}$  respectively. Denoting by $H_{\circ}$ the restriction of $H$ to the domain $D(H_{\circ}):=\{\Psi\in H^{1}(\RE):\Psi(y)=0\}$,  one has that $H_{\circ}$ is closed symmetric, has defect indices $(2,2)$ and  adjoint $H_{\circ}^{*}=H_{-}\oplus H_{+}$. In order to define self-adjoint extensions of $H_{\circ}$ we consider Hermitian $2\times2$ matrices 
\begin{equation*}
A = \begin {bmatrix}
\al_1 &\gamma \\ 
\bar \gamma &\al_2
\end {bmatrix},\qquad \al_1,\al_{2} \in\RE\,,\ \gamma \in \CO\,.
\end{equation*}
Then one gets a self-adjoint operator 
$H_{A}$ on $L^2(\mathbb{R})\otimes\mathbb{C}^2$ by restricting $H_{-}\oplus H_{+}$ to the
domain
\begin{align}
D(H_{A})=&
\left\{\Psi\in H^{1}({\mathbb R}\backslash\{y\})\otimes\CO^{2}: ic\,\sigma_{1}[\Psi]_{y} =A\q\right\},  \label{DHAy2}
\end{align}
where $H^{1}({\mathbb R}\backslash\{y\}):=H^{1}(-\infty,y)\oplus H^{1}(y,+\infty)$, 
 \[
[\Psi]_{y}=\begin {pmatrix} [\psi_1]_{y}\\ [\psi_2]_{y}
 \end {pmatrix}:=\Psi(y^{+})-\Psi(y^{-})
\]
denotes the jump of $\Psi$ at the point $y$ and 
\begin{equation*}
\q=\begin {pmatrix} q_{1}\\ q_{2} \end {pmatrix}
:=\frac12\,\big(\Psi(y^{+})+\Psi(y^{-})\big)
\end{equation*} 
denotes the mean value of $\Psi$ at the point $y$, and is also called the {\it charge} of the wave vector $\Psi$.
The case $A=0$ gives the free Dirac  operator $H$. By using distributional derivatives one has 
\begin{align}\label{GS}
H_{A}\Psi=D_{m}\Psi+i\hbar\, c\, \sigma_{1}[\Psi]_{y}\delta_{y}
= 
D_{m}\Psi+\hbar A\q\,\delta_{y}\,.
\end{align}
The domain and the action of $H_{A}$ can be described in an alternative way as follows (for simplicity of exposition we consider only the case where $m>0$, a similar description holds also in the $m=0$ case): let $G$ denote the solution of   $-D_{m}G=\delta_{y}\otimes\II$, i.e. 
\[
G(x)=-\frac1{2\hbar c}\,e^{-\frac{mc}{\hbar}|x-y|}(i\,\sgn(x-y)\sigma_{1} +\sigma_{3})
\,.
\]
Then, since 
\[
{i}{\hbar c}\,\sigma_{1}[G\underline\xi]_{y}=\underline\xi\,,
\]
and $\Psi\in H^{1}(\RE\backslash\{y\})\otimes\CO^{2}$ belongs to $H^{1}(\RE)\otimes\CO^{2}$ if and only if $[\Psi]_{y}=0$, one gets
\[
H^{1}(\RE\backslash\{y\})\otimes\CO^{2}=\{\Psi=\Phi+G\underline\xi\,,\ \Phi\in H^{1}(\RE)\otimes\CO^{2}\,,\ \xi\in\CO^{2}
\}\,,
\]
\[
H_{\circ}^{*}\Psi=H\Phi
\]
and so, since 
\begin{equation}\label{Gmedxi}
\frac12\,\big(G\underline\xi(y^{+})+G\underline\xi(y^{-})\big)=-\frac{\sigma_{3}\underline\xi}{2\hbar c}\,,
\end{equation}
the self-adjoint extension $H_A$ can be equivalently defined as 
\begin{equation}\label{dominio}
D(H_{A})=\left\{\Psi=\Phi+G\underline\xi\,,\ \Phi\in H^{1}(\RE)\otimes\CO^{2}\,,\ \xi\in\CO^{2}\,,\ \left(\II+\frac1{2c}\,A\sigma_{3}\right)\underline\xi=\hbar A\Phi(y)\right\}\,,
\end{equation}
\[
H_{A}\Psi=H\Phi\,.
\]
We now consider the Cauchy problem
\be\label{cauchylin}
\begin{cases}
i\hbar\frac{d}{d t}\Psi(t)=H_{A}\Psi(t)&\\
\Psi(0)=\Psi_{\circ}\,.&
\end{cases}
\ee
Since $H_{A}$ is self-adjoint, such a Cauchy problem is well-posed for any $\Psi_{\circ }\in L^2(\mathbb{R})\otimes\mathbb{C}^2$ by Stone's theorem. In the following proposition we  give a representation formula for the solution of problem \eqref{cauchylin} in the case $\Psi_{\circ}\in D(H_{A})$. For simplicity of exposition we only consider the case $t\ge 0$; a similar representation holds for $t\le 0$. 
\begin{proposition}\label{p:cauchylin}
Let  $\Psi_\circ \in D(H_{A})$. Then for any $t\geq 0$ the solution  $\Psi(t)=e^{-\frac i\hbar tH_{A}} \Psi_\circ$ of the Cauchy problem \eqref{cauchylin} is given by
\begin{equation}\label{Psilin}\begin{aligned}
\Psi(x,t)= & \Psi^{f}(x,t)  -\frac{i}{2c}\theta\left( t-\frac{|x-y|}c \right) \begin {bmatrix} 1& \sgn(x-y) \\ \sgn(x-y) &1  \end {bmatrix} A\q\left( t-\frac{|x-y|}c \right) \\ &-i\theta\left( t-\frac{|x-y|}c \right)\int_0^{t-\frac{|x-y|}{c}} ds\,K(x-y,t-s) A \q(s)\,,
\end{aligned} \end{equation}
where  $\Psi^{f}(t):=e^{-\frac i\hbar tH} \Psi_\circ$ and   
$\q(t) 
$ is the solution of the integral equation 
\begin{equation}\label{qlin}
\q(t) = \Psi^{f}(y,t) -\frac{i}{2c}\, A \q(t)- i \int_0^t ds\,K(0,t-s) A\q(s) \,.
\end{equation}
Here $\theta$ denotes Heaviside's step function and the matrix-valued kernel $K$ is defined by
\[
K(x,t) =- \frac{mc}{2\hbar} \left( i\sigma_{3}J_0\left(\frac{m\,c}{\hbar}\sqrt{(ct)^2-x^2}\,\right)+
(ct\II+x\sigma_{1})\ \frac{J_1\left(\frac{m\,c}{\hbar}\sqrt{(ct)^2-x^2}\,\right)}{\sqrt{(ct)^2-x^2}}\right)\,,
\]
$J_{k}$ denoting the Bessel function of order $k$.
\end{proposition}
\begin{proof} Recall that $t\mapsto e^{-\frac{i}{\hbar} t H_{A}}$ is a strongly continuous unitary group 
and moreover note that the maps on $H^1(\RE\backslash\{y\})$ to $\CO$
\[f \mapsto \lim_{x\to y^\pm} f(x)\equiv f(y^\pm)\,, \qquad f\in H^1(\RE\backslash\{y\})\] 
are continuous. Then the map $t\mapsto \q(t)$, with $\q(t)$ defined as 
\begin{equation}\label{q}
\q(t) := \frac12(\Psi(y^+,t) +\Psi(y^-,t)),
\end{equation} 
is continuous as well. 

The relation \eqref{GS} leads us to consider the distributional Cauchy problem 
\begin{equation}\label{dcp}
\begin{cases}
i\hbar\frac{d}{d t}\Psi(t) =D_{m}\Psi(t)+\hbar A\q(t)\delta_{y}&
\\
\Psi(0)=\Psi_{\circ}\,,
\end{cases}
\end{equation}
where  $\q(t)$ is defined  by \eqref{q}. 
Then 
\[
\Psi(t)=\Psi^{f}(t)+\Psi^{\delta}(t)\,,
\]
where $\Psi^f(t) = e^{-\frac{i}{\hbar} t H }\Psi_\circ $, and  $\Psi^{\delta}(t)
$ solves \eqref{dcp} with zero initial conditions. 
By  Duhamel's formula
\[
\Psi^{\delta}(t)=-i\int_0^{t}ds\,e^{-\frac i\hbar (t-s)H}Aq(s)\delta_{y}\,. 
\]
Let us notice that, by \eqref{free}, the group of evolution $\exp({-\frac i\hbar tH})$ continuously maps $\S(\RE)\otimes\CO^{2}$ in  $\S(\RE)\otimes\CO^{2}$ and so it extends by duality to a group of evolution (which we denote by the same symbol) on $\S'(\RE)\otimes\CO^{2}$ to  $\S'(\RE)\otimes\CO^{2}$. Using the  definition of the unitary group $e^{-\frac{i}{\hbar}t H}$, see Eq. \eqref{Ufree}, we get 
\[\begin{aligned}
\Psi^\delta(t) =&  -\frac{i}{2} \int_0^t ds\left( (\mathbb{1}+\sigma_1)Aq(s) \, \delta_{y+c(t-s)} + (\mathbb{1}-\sigma_1) Aq(s) \,\delta_{y-c(t-s)}\,
\right) \\ 
&-i \int_0^t ds \int_{-c(t-s)}^{c(t-s)} d\xi \, K(\xi,t-s) Aq(s) \, \delta_{y+\xi}\,  .
\end{aligned}\] 
Exploiting the Dirac-delta distributions in the integrals (recalling that $t\geq0$),  we get 
\begin{equation}\begin{aligned}\label{Pside}
\Psi^\de(x,t)= &   -\frac{i}{2c}\theta\left( t-\frac{|x-y|}c \right) \begin {bmatrix} 1& \sgn(x-y) \\ \sgn(x-y) &1  \end {bmatrix} Aq\left( t-\frac{|x-y|}c \right) \\ &-i\theta\left( t-\frac{|x-y|}c \right)\int_0^{t-\frac{|x-y|}{c}} ds\,K(x-y,t-s) A q(s).
\end{aligned}\end{equation}
Therefore we have proved that the solution of the Cauchy problem \eqref{cauchylin} satisfies the identity \eqref{Psilin} with $\q(t)$ defined by \eqref{q}. 

To prove that $\q(t)$ satisfies the identity \eqref{qlin} we note that,  by equation \eqref{Psilin}, 
one has 
\begin{equation}\label{4.8}
\Psi(y_\pm,t) = \Psi^f(y,t)  -\frac{i}{2c}(\mathbb{1}\pm\sigma_1) Aq(t) -i \int_0^{t} ds\,K(0,t-s) A q(s). 
 \end{equation}
\end{proof}

\begin{remark}
Note that from   Eq. \eqref{Psilin} one can see that $\Psi(t)$ satisfies the boundary conditions in \eqref{DHAy2}  for any $t>0$. Indeed, by Eq. \eqref{4.8}, one has that 
 \[
\left[\Psi(t)\right]_{y}= -\frac{i}{2c}((\mathbb{1}+\sigma_1)-(\mathbb{1}-\sigma_1))A\q(t)  = -  \frac{i}{c}\,\sigma_{1}A\q(t).
 \]
\end{remark}

\begin{remark}
Despite the presence of the $\theta$-function at the r.h.s. of Eq. \eqref{Psilin}, the function $\Psi(x,t)$ is continuous in $x = y\pm ct$. In fact, from formula  \eqref{Ufree}, it follows that   for any $t>0$, $\Psi^f(t)$ is discontinuous  in $x=y\pm ct$ and  
\begin{equation}\label{jump1}
[\Psi^f(t)]_{y\pm ct} =\frac12[(\mathbb{1}\pm\sigma_1)\Psi_\circ]_y = \frac12 \begin {pmatrix}
[\psi^\circ_1]_y \pm [\psi^\circ_2]_y \\  \pm [\psi^\circ_1]_y + [\psi^\circ_2]_y
\end {pmatrix}.
\end{equation}
On the other hand, since 
\[\lim_{x\to(y\pm c t)^\mp} \theta \left( t-\frac{|x-y|}c \right) =1 \quad ;\qquad\lim_{x\to(y\pm c t)^\pm} \theta \left( t-\frac{|x-y|}c \right)  =0,
\] one has that 
\[\left[ \theta \left( t-\frac{|x-y|}c \right)\right]_{x=y\pm ct} = \lim_{x\to(y\pm c t)^+} \theta \left( t-\frac{|x-y|}c \right) - \lim_{x\to(y\pm c t)^-} \theta \left( t-\frac{|x-y|}c \right)  = \mp 1.\]
Then, Eq. \eqref{Pside} gives 
\begin{equation}\label{jump2}
[\Psi^\de(t)]_{y\pm ct} = \pm \frac{i}{2c} (\mathbb{1}\pm\sigma_1) Aq(0) =\mp \frac12 (\mathbb{1}\pm\sigma_1) \begin {pmatrix}
[\psi^\circ_2]_y \\ [\psi^\circ_1]_y \end {pmatrix} =  - \frac12 \begin {pmatrix}
[\psi^\circ_1]_y \pm [\psi^\circ_2]_y \\  \pm [\psi^\circ_1]_y + [\psi^\circ_2]_y
\end {pmatrix},\end{equation}
where in the second equality we used the fact that $\Psi_\circ \in D(H_{A})$. Eqs. \eqref{jump1} and \eqref{jump2} give $[\Psi(t)]_{y\pm ct} = [\Psi^f(t)]_{y\pm ct} + [\Psi^\de(t)]_{y\pm ct} =0 $. 
\end{remark}
\begin{remark}\label{r:Psifyt}
From Eq. \eqref{Ufree}, it follows that,  for any $t>0$, $\Psi^f(y,t)$ is continuous in $t$, and 
\[\Psi^f(y,0) \equiv \lim_{t\to0} \Psi^f(y,t) = \frac12\big((\mathbb{1}+\sigma_1)\Psi_\circ(y^-) +(\mathbb{1}-\sigma_1)\Psi_\circ(y^+)\big). \]
This implies that 
\[
\Psi^f(y,0)  -\frac{i}{2c} A\q(0) =  \frac12\big((\mathbb{1}+\sigma_1)\Psi_\circ(y^-) +(\mathbb{1}-\sigma_1)\Psi_\circ(y^+)\big) + \frac12 \,\sigma_{1} [\Psi_\circ]_y = \q(0),
\]
which is in agreement with the fact that $\q(t)$ satisfies Eq. \eqref{qlin}. 
\end{remark}

%

\section{The Cauchy problem for the Dirac equation with concentrated non-linearity.}
Now we define a Dirac operator $H^{nl}_{A}$ with concentrated nonlinearity such that the coupling between the jump and the mean value of the spinor-function is given by a nonlinear relation. To this aim we define the nonlinear domain 
\begin{equation*}
D(H^{nl}_{A})  = 
\left\{\Psi \in \mathcal{H} :\Psi \in H^{1}(\RE\backslash\{y\}) \otimes \CO^{2},\,ic \,\sigma_{1}[\Psi ]_{y} = A(\q) \q  \right\}, 
\end{equation*}
where $\CO^2\ni \underline z\mapsto A(\underline z)$ is a  matrix-valued  function   such that $A(\underline z)$ is self-adjoint for all $\underline z$; $H^{nl}_{A}$ is then defined as the restriction of $H^{*}_{\circ}=H_{-}\oplus H_{+}$ to $D(H^{nl}_{A})$, so that
\be\label{Hnl}
H^{nl}_{A}:D(H^{nl}_{A})\subset L^{2}(\RE)\otimes\CO^{2}\to L^{2}(\RE)\otimes\CO^{2}\,,
\quad H^{nl}_{A}\Psi=D_{m}\Psi+\hbar A(\q)\q\,\delta_{y}\,.
\ee
 \begin{remark} We use the notation $H^{nl}_{A}$ just for convenience, indeed the nonlinear operator $H^{nl}_{A}$ depends on the function $\z\mapsto A(\z)\z$ and not only on $A(\z)$: there could be two different matrices $A_{1}(\z)$ and $A_{2}(\z)$ such that $A_{1}(\z)\z=A_{2}(\z)\z$.
 Clearly $H^{nl}_{A}=H_{A}$ whenever $A$ is $\underline z$-independent. 
 \end{remark}
 In order to solve the nonlinear Cauchy problem
\be\label{cauchynonlin}
\begin{cases}
i\hbar\frac{d}{d t}\Psi(t)=H^{nl}_{A}\Psi(t)&\\
\Psi(0)=\Psi_{\circ}\,,&
\end{cases}
\ee
we mimic the Dirac flow in the representation formula  given in Prop. \ref{p:cauchylin}. 

Take $\Psi_\circ \in D(H^{nl}_{A})$ and set $\Psi^f (t) = e^{-\frac{i}{\hbar} t H} \Psi_\circ $. For  $t\geq0$, the non-linear Dirac flow $U_{t}^{nl}$ 
is defined by $U^{nl}_{t}\Psi_\circ:=\Psi(t)$, where  
\begin{equation}\label{Psinl}\begin{aligned}
\Psi(x,t):= & \Psi^{f}(x,t)  -\frac{i}{2c}\theta\left( t-\frac{|x-y|}c \right) \begin {bmatrix} 1& \sgn(x-y) \\ \sgn(x-y) &1  \end {bmatrix} \left(A(\q) \q\right)\left( t-\frac{|x-y|}c \right) \\ &-i\theta\left( t-\frac{|x-y|}c \right)\int_0^{t-\frac{|x-y|}{c}} ds\,K(x-y,t-s) (A(\q) \q)(s)\,,
\end{aligned} \end{equation} 
where $\q(t)$ is the solution of the nonlinear integral equation 
\begin{equation}\label{qnl}
\q(t) = \Psi^f(y,t) -\frac{i}{2c}\,( A(\q) \q)(t)- i \int_0^t ds\,K(0,t-s) (A(\q) \q)(s) ,\end{equation}
and  $(A(\q)\q)(t)$ is a shorthand notation for  $A(\q(t))\q(t)$.

The first step in order to prove the well-posedness  of the problem \eqref{cauchynonlin} is to show that, for any $T>0$, Eq. \eqref{qnl} admits a unique (sufficiently regular) solution for $t\in[0,T]$.  This is achieved in   Lemma \ref{l:qwp} below. 
\par
In the proof of Lemma \ref{l:qwp}, we need the map 
\[
F_{A}:\CO^2 \to \CO^2\,,\quad
F_{A}(\underline z) : = \underline z + \frac{i}{2c}\, A(\underline z)\underline z
\]
to be locally bi-Lipschitz continuous. Therefore, we make the following assumption on the matrix valued function  $A(\underline z)$:
\begin{assumption}\label{a:F}
The map $\z\mapsto A(\z)$ from $\CO^2$ to the space  of $2\times2$ self-adjoint matrices is such that  $F_{A}$ is $C^1$-diffeomorphism as a map on $\RE^{4}$ to itself.
\end{assumption}
By Hadamard's  global inverse function theorem (see e.g. \cite{Gordon:1972kh, Gordon:1973iq} and references therein), a $C^1$ map $\Phi:\RE^N \to \RE^N$ is a $C^{1}$-diffeomorphism if and only its Jacobian determinant never vanishes and $\|\Phi(\underline x)\|\to+\infty$ as $\|\underline x\|\to \infty$.  Since the complex map $F_{A}$ can be equivalently seen as a map from $\RE^4$ to $\RE^4$,  such a global inverse function theorem applies to $F_{A}$ as well. By $|F_{A}(\underline z)|^2 = |\underline z|^2+ |A(\underline z)\underline z|^2/(4c^2)$, it follows  $|F_{A}(\underline z)|\to+\infty$ whenever $|\underline z|\to+\infty$. Hence, Assumption \ref{a:F} is equivalent to
\vskip8pt\noindent
{\bf Assumption \theassumption'}
The map $\z\mapsto A(\z)$ from $\CO^2$ to the space of $2\times2$ self-adjoint matrices 
is such that  $F_{A}$ is $C^1(\RE^{4},\RE^{4})$ and its Jacobian determinant never vanishes.
\vskip8pt
We are now ready to state our first results that concerns the well-posedness of the equation for $\q(t)$. 
\begin{lemma}\label{l:qwp} Let $A(\underline z)$ be such that Assumption \ref{a:F} is satisfied. Then for any $\Psi_\circ\in D(H^{nl}_{A})$ and  $T>0$ there exists a unique solution $\q\in H^{1}(0,T)\otimes \CO^2$ of Eq. \eqref{qnl}.
\end{lemma}
\begin{proof} At first we prove that there exists a unique solution $\q\in C[0,T]\otimes \CO^2$. We equivalently show that, for fixed  $T>0$ and $\Psi_\circ\in D(H^{nl}_{A})$,  there exists $\bar t>0$ for which the following holds true: $\forall k\ge 0$ such  that $k\bar t\le T$, suppose that there exists an unique solution $\q^{k}\in C[0,k\bar t\,]\otimes\CO^{2}$ of Eq. \eqref{qnl} (no assumption in the case $k=0$); then Eq. \eqref{qnl} has an unique solution $\q^{k+1}\in C[0,(k+1)\bar t\,]\otimes\CO^{2}$.\par
To begin with, we show that , if $\q(t)$ solves Eq. \eqref{qnl}, then   there exists a positive constant $C_1$ such that \begin{equation}\label{apriori}
\sup_{t\in[0,T]}|(A(\q)\q)(t)| \leq C_1. 
\end{equation}
To prove this claim  recall that, by Rem. \ref{r:Psifyt},  $ \sup_{t\in[0,T]}|\Psi^f(y,t)|\leq C$, and that  $\sup_{t\in[0,T]}|K(0,t)| \leq C $ for some positive constant $C$. Hence, by Eq. \eqref{qnl},  and using the fact that $A(\underline z)$ is self-adjoint, we get that for all $t\in [0,T]$ the following inequality  holds true
\[|(A(\q)\q)(t)| \leq 2c \left|\left(\II +\frac{i}{2c}\, A(\q(t))\right)\q(t)\right| \leq C \left(1 +\int_0^t ds\, |(A(\q)\q)(s)|\right).  \] 
Then the bound \eqref{apriori} follows by Gr\"onwall's inequality.

Next, let us pose $t_{k}:=k\bar t$.
For  $t\in [t_k,t_k+\bar t\,]$, the solution of Eq. \eqref{qnl} satisfies the identity  
\begin{equation}\label{contr0}
\left(\II +\frac{i}{2c} A(\q(t))  \right)\q(t) = \Psi^f(y,t) - i \int_0^{t_k} ds\,K(0,t-s) (A(\q^k) \q^k)(s) - i \int_{t_k}^t ds\,K(0,t-s) (A(\q) \q)(s). \end{equation}
 We set 
 \[
 f_k(t) = \Psi^f(y,t) - i \int_0^{t_k} ds\,K(0,t-s) (A(\q^k) \q^k)(s) \] and\[ I_k(\q)(t) = - i \int_{t_k}^t ds\,K(0,t-s) (A(\q) \q)(s),
 \]
and rewrite Eq. \eqref{contr0} as 
\[
\q(t) = F_{A}^{-1}(f_k(t) + I_k(\q)(t)), 
\]
where by Assumption \ref{a:F}, $F_{A}^{-1}$ exists and is a $C^1$ map from $\CO^2$ to $\CO^2$. 

Since $\Psi^f(y,t) $  and $K(0,t)$ are bounded, and by the bound \eqref{apriori}, it follows that 
\[\sup_{t\in [0,T]}|f_k(t)|  \leq C(1+T C_1)\equiv R_1.     \]
Let $B_k(R):= \{\underline{g}\in C[t_k,t_k+\bar t\,]\otimes\CO^2:\, \sup_{t\in [t_k,t_k+\bar t\,]}|\underline{g}(t)|\leq R\}$.

 For any $\underline{g}\in B_k(2R_1)$, and $\bar t\leq  t_1 = R_1(C \sup_{|\underline{z}|\leq 2R_1} |A(\underline{z})\underline{z}|)^{-1}$,  independent on $k$, we have that 
 \[\sup_{t\in[t_k,t_k+\bar t\,]}|f_k(t) + I_k(\underline{g})(t) | \leq R_1 + \bar t C \sup_{|\underline{z}|\leq 2R_1} |A(\underline{z})\underline{z}|\leq 2R_1. \]

Define the map 
\[
G_k(\underline{g}) :=  F_{A}^{-1}(f_k + I_k(\underline{g})).
\]
The map $G_k$ is continuous in $B_k(2 R_1)$. The self-adjointness of $A(\underline{z})$, implies that $|F_{A}^{-1}(\underline{z})|\leq |\underline{z}|$, hence for $\bar t \leq t_1$ one has that 
\[\sup_{t\in[t_k,t_k+\bar t\,]}|G_k(\underline{g})(t)| \leq  \sup_{t\in[t_k,t_k+\bar t\,]}|f_k(t) + I_k(\underline{g})(t) | \leq 2R_1 \] which  means that $G_k$ maps $B_k(2 R_1)$ into itself.

By Assumption \ref{a:F}, we have that the maps $F_{A}^{-1}(\underline z)$ and $A(\underline{z})\underline{z}=i2c(\underline z-F_{A}(\underline z))$ are locally Lipschitz. More precisely, for any $\underline{z}_1$ and $\underline{z}_2$ such that $|\underline{z}_1|,|\underline{z}_2|\leq R$, there exist two constants $\kappa_F(R)$ and $\kappa_A(R) $ such that
\[|F_{A}^{-1}(\underline{z}_1)- F_{A}^{-1}(\underline{z}_2)|\leq \kappa_F(R)|\underline{z}_1-\underline{z}_2|\quad \text{and} \quad|A(\underline{z}_1)\underline{z}_1- A(\underline{z}_2)\underline{z}_2|\leq \kappa_A(R)|\underline{z}_1-\underline{z}_2|.\] 
Take $\bar t \leq t_1$. For any  $\underline{g}_1, \underline{g}_2 \in B_k(2R_1)$,  one has that $\sup_{t\in[t_k,t_k+\bar t\,]}|f_k(t) + I_k(\underline{g}_j)(t) | \leq 2R_1$, and  
\[\begin{aligned}
\sup_{t\in[t_k,t_k+\bar t\,]} |G_k(\underline{g}_1(t))- G_k(\underline{g}_2(t))| & \leq \kappa_F(2R_1)\sup_{t\in[t_k,t_k+\bar t\,]}|I_k(\underline{g}_1)(t) - I_k(\underline{g}_2)(t)|  \\ 
& \leq \kappa_F(2R_1)\kappa_A(2R_1)C \bar t  \sup_{t\in[t_k,t_k+\bar t\,]}|\underline{g}_1(t) - \underline{g}_2(t)|. 
\end{aligned}\]
Set $ t_2 =   (2 \kappa_F(2R_1)\kappa_A(2R_1)C)^{-1}$, independent on $k$, and  $\bar t =\min\{t_1,t_2\}$, then  the map $G_k$ is a contraction in $B_k(2R_1)$. By the Banach-Caccioppoli fixed point theorem, this implies that there exists a unique solution $\q^*(t)\in B_k(2R_1)$ of Eq. \eqref{contr0}. 

By construction the function $\q^{k+1}(t)$ which is equal to $\q^k(t)$ for $t\in [0,t_k]$, and to $\q^*(t)$ for $t\in [t_k,t_k + \bar t\,]$,  is indeed in $C[0,t_k+\bar t\,]\otimes\CO^{2}$ and  solves Eq. \eqref{qnl} for $t\in[0,t_k+\bar t\,]$. \par
By 
$$
\q(t)=F_{A}^{-1}\left(\Psi^f(y,t) - I(t) \right)\,,\quad I(t):=i \int_0^t ds\,K(0,t-s) (A(\q) \q)(s)\,,
$$
since $\Psi^f(y,\cdot)\in H^{1}(0,T)\otimes\CO^{2}$ (see Prop. \ref{H1} in Appendix \ref{app:A}), $I\in C^{1}[0,T]\otimes\CO^{2}$ and $F_{A}^{-1}$ is Lipschitz continuous, in conclusion $\q\in H^{1}(0,T)\otimes\CO^{2}$.
 \end{proof}
By the previous results, the proof of global well-posedness of the Cauchy problem  \eqref{cauchynonlin} follows:
\begin{theorem}\label{t:main}
Let $A(\underline z)$ be such that Assumption \ref{a:F} is satisfied. Then for any $\Psi_\circ\in D(H^{nl}_{A})$ the formulae \eqref{Psinl} and \eqref{qnl} provide the unique, global in time, solution $\Psi(t)$ of the Cauchy problem \eqref{cauchynonlin}; more precisely, $\Psi\in C^{1}(\RE_{+},L^{2}(\RE)\otimes \CO^{2})$, $\Psi(t)\in D(H^{nl}_{A})$ and \eqref{cauchynonlin} holds for any $t\ge 0$. 
\end{theorem}
\begin{proof} By \eqref{Hnl} and by the same reasonings as in the linear case provided in Section \ref{s:lin} (replacing $A\q$ with $A(\q)\q$), one has that $\Psi(t)$ given in formula \eqref{Psinl} 
solves the distributional Cauchy problem 
\[
\begin{cases}
i\hbar\frac{d}{d t}\Psi(t) =D_{m}\Psi(t)+\hbar (A(\q)\q)(t)\,\delta_{y}&
\\
\Psi(0)=\Psi_{\circ}\,,
\end{cases}
\]and, since $\q(t)$ solves \eqref{qnl}, one gets $\Psi(t)\in D(H^{nl}_{A})$ for any $t\ge0$. Therefore, to conclude the proof we need to show that the map $t\mapsto \Psi(t)$ belongs to $C^{1}(\RE_{+},L^{2}(\RE)\otimes \CO^{2})$. Since $\Psi(t)\in D(H^{nl}_{A})\subset H^{1}(\RE\backslash\{y\})$, we have the decomposition (in the following we suppose $m>0$; similar considerations hold in the $m=0$ case)
\begin{equation}\label{dec}
\Psi(t)=\Phi(t)+G\underline\xi(t)\,,\quad\Phi(t)\in H^{1}(\RE)\otimes\CO^{2}\,,\quad\underline\xi(t)=i\hbar c\,\sigma_{1}[\Psi(t)]_{y}\,.
\end{equation}
Let us notice that, since \begin{equation}\label{xiq}\underline\xi(t)=\hbar (A(\q)\q)(t),\end{equation} and $\underline z\mapsto A(\underline z)\underline z$ is Lipschitz continuous,  $t\mapsto \underline\xi(t)$ belongs to $ H^{1}(0,T)\otimes\CO^{2}$ for any $T>0$ by Lemma \ref{l:qwp}. Moreover, since 
\begin{equation}\label{HAnl}
H^{nl}_{A}\Psi=D_{m}(\Phi+G\underline\xi(t))+\underline\xi(t)\,\delta_{y}=H\Phi\,,
\end{equation} one has that $\Phi(t)$ solves the Cauchy problem
\begin{equation}\label{Creg}
\begin{cases}
i\hbar\frac{d}{d t}\Phi(t)=H\Phi(t)-i\hbar\, G\underline{\dot\xi}(t)&\\
\Phi(0)=\Phi_\circ\,,&
\end{cases}
\end{equation}
with $\Phi_\circ:=\Psi_{\circ}-G\underline{\xi}(0)\in H^{1}(\RE)\otimes\CO^{2}$. Therefore 
\[
\Psi(t)=e^{-\frac {i}{\hbar}tH}\Phi_\circ-\int_{0}^{t}ds\,e^{-\frac {i}{\hbar}(t-s)H}G\underline{\dot\xi}(s)+G\underline{\xi}(t)\,.
\]
Since $t\mapsto e^{-\frac {i}{\hbar}tH}\Phi_\circ$ belongs to $C^{1}(\RE_{+},L^{2}(\RE)\otimes \CO^{2})$ and $-HG\underline{\dot\xi}=\underline{\dot\xi}\delta_{y}$, to conclude we need to show that the map 
\[
t\mapsto \Upsilon(t):=\frac{d}{d t}\left(\int_{0}^{t}ds\,e^{-\frac {i}{\hbar}(t-s)H}G\underline{\dot\xi}(s)-G\underline{\xi}(t)\right)
=
\frac{i}{\hbar}\int_{0}^{t}ds\,e^{-\frac {i}{\hbar}(t-s)H}\underline{\dot\xi}(s)\delta_{y}
\]
belongs to $C(\RE_{+},L^{2}(\RE)\otimes \CO^{2})$. 
By the same calculations that led to \eqref{Pside}, one gets  
\[
\Upsilon(t)=\frac{i}{\hbar}\left(\frac{1}{2c}\,\Upsilon_{1}(t)+\Upsilon_{2}(t)\right)\,,
\] 
where
\[
\Upsilon_{1}(x,t)= \theta\left( t-\frac{|x-y|}c \right) \begin {bmatrix} 1& \sgn(x-y) \\ \sgn(x-y) &1  \end {bmatrix} \underline{\dot\xi}\left( t-\frac{|x-y|}c \right) \,,
\]
\[ 
\Upsilon_{2}(x,t)=\theta\left( t-\frac{|x-y|}c \right)\int_0^{t-\frac{|x-y|}{c}} d\tau\,K(x-y,t-\tau) \underline{\dot\xi}(\tau)\,.
\]
Let $\kappa$ denote the bound for the kernel $K$:
\[\kappa := \max_{i,j = 1,2} \sup_{0<t<t_\circ, |x|< ct} |K_{i,j}(x,t)|. \]
One has (supposing $0\le s<t\le t_{\circ}$, the same kind of reasonings hold in the case $0\le t<s\le t_{\circ}$ )
\begin{align*}
\|\Upsilon_{1}(t)-\Upsilon_{1}(s)\|^{2}\leq&C\left(\int_{0}^{cs}dx\,\left|\underline{\dot\xi}\left(t-\frac{x}c\right)-\underline{\dot\xi}\left(s-\frac{x}c\right)\right|^{2}+\int_{cs}^{ct}dx\,\left|\underline{\dot\xi}\left(t-\frac{x}c\right)\right|^{2} \right)\\
\le&C\left(\int_{0}^{t_{\circ}}dx\,\left|\underline{\dot\xi}\left(t-s+x\right)-\underline{\dot\xi}\left(x\right)\right|^{2}+\int_{0}^{t-s}dx\,\left|\underline{\dot\xi}(x)\right|^{2} \right) 
\end{align*}
and 
\begin{align*}
&\|\Upsilon_{2}(t)-\Upsilon_{2}(s)\|^{2}\\
\leq & 2\int_{|x|\leq cs}dx\left|
\int_{0}^{t-|x|/c}d\tau\, K(x,t-\tau) \underline{\dot\xi}(\tau) - \int_0^{s-|x|/c}d\tau\,K(x,s-\tau)\underline{\dot\xi}(\tau)\right|^{2}\\
&+2\int_{cs\leq|x|\leq ct}dx\left|\int_{0}^{t-|x|/c}d\tau\,K(x,t-\tau)\underline{\dot\xi}(\tau)\right|^{2}
\\
\leq & 4\int_{|x|\leq cs}dx\left(
\int_{|x|/c}^{s}d\tau\, \left|  K(x,\tau)\left(  \underline{\dot\xi}(t-\tau) - \underline{\dot\xi}(s-\tau)\right) \right|\right)^{2} 
\\
&+
4\int_{|x|\leq cs}dx\left(
\int_{s}^{t}d\tau\, \left| K(x,\tau) \underline{\dot\xi}(t-\tau) \right|\right)^{2}
\\
&+2\int_{cs\leq|x|\leq ct}dx\left(\int_{0}^{t-|x|/c}d\tau\,\left| K(x,t-\tau)\underline{\dot\xi}(\tau) \right|\right)^{2}
\\
\leq & 16\kappa^2 \int_{|x|\leq cs}dx\left(
\int_{|x|/c}^{s}d\tau\, \left|   \underline{\dot\xi}(t-\tau) - \underline{\dot\xi}(s-\tau) \right|\right)^{2} +
16\kappa^2 \int_{|x|\leq cs}dx\left(
\int_{s}^{t}d\tau\, \left| \underline{\dot\xi}(t-\tau) \right|\right)^{2}
\\
&+8\kappa^2\int_{cs\leq|x|\leq ct}dx\left(\int_{0}^{t-|x|/c}d\tau\,\left|\underline{\dot\xi}(\tau) \right|\right)^{2}
\\
\leq & C\int_{0}^{s}d\tau\, \left|   \underline{\dot\xi}(t-\tau) - \underline{\dot\xi}(s-\tau) \right|^2 +C(t-s) \|\underline{\dot\xi}\|^{2}_{H^{1}(0,t_{\circ})}
\\
\leq & C\int_{0}^{t_\circ}d\tau\, \left|   \underline{\dot\xi}(t-s+\tau) - \underline{\dot\xi}(\tau) \right|^2 +C(t-s) \|\underline{\dot\xi}\|^{2}_{H^{1}(0,t_{\circ})}
\end{align*}
Since 
\[
\lim_{s\to t}\int_{0}^{t_{\circ}}dx\,\left|\underline{\dot\xi}\left(t-s+x\right)-\underline{\dot\xi}\left(x\right)\right|^{2}=0
\]
(by $\underline\xi\in H^{1}(0,T)$ and by the continuity of the shift operator, see e.g. \cite[page 11]{Sobolev:2010wz}) and 
\[\|\Upsilon(t)-\Upsilon(s)\| \leq C \left( \|\Upsilon_1(t)-\Upsilon_1(s)\|+\|\Upsilon_2(t)-\Upsilon_2(s)\|\right)\]
we conclude 
\[
\lim_{s\to t}\|\Upsilon(t)-\Upsilon(s)\|=0\,.
\]
\end{proof}
\begin{remark}\label{equiv} The Dirac differential operator  $D_{m}$ has many different equivalent representations: given any unitary map $U:\CO^{2}\to\CO^{2}$, one defines an equivalent Dirac operator by $\tilde D_{m}:=(\II\otimes U^{*})D_{m}(\II\otimes U)$, i.e. 
$$
\tilde D_{m}\Psi=-i\hbar\,c\,\tilde\sigma_{1}\frac{d\Psi}{dx}+m\, c^{2}\tilde\sigma_{3}\Psi\,,
\qquad \tilde \sigma_{k}:=U^{*}\sigma_{k}U\,.
$$ 
The relation between the corresponding non-linear operators is 
$$
(\II\otimes U^{*})H^{nl}_{A}(\II\otimes U)\Psi=\tilde H^{nl}_{\tilde  A}\Psi:=\tilde  D_{m}\Psi+\hbar\tilde  A(\q)\q\,\delta_{y}\,,\qquad\tilde  A(\z):=U^{*}A(U\z)U\,.
$$ 
Since
\[
F_{\tilde A}(\z)=U^{*}\left(\II+\frac{i}{2c}\, A(U\z)\right)U\z=U^{*} F_{A}(U\z)\,,
\]
$F_{A}$ satisfies Assumption \ref{a:F} if and only if $F_{\tilde A}$ does. This shows that our global well-posedness  result holds in any representation and the Assumption \ref{a:F} is an invariant one.
\end{remark}

\begin{theorem}[Mass conservation]\label{t:mass}Let $\Psi_\circ \in D(H_A^{nl})$, then the $L^2$-norm is conserved along the flow associated to the  Cauchy problem \eqref{cauchynonlin}. 
\end{theorem}
\begin{proof}
We take the derivative 
\[\frac{d}{dt} \|\Psi(t)\|^2 = 2 \Re\left\langle \dot \Psi(t) , \Psi(t)\right\rangle. \]
We write  $\Psi(t)$ as in Eq. \eqref{dec} and  use Eq. \eqref{Creg}  to get 
\begin{equation}\label{boh}
\left\langle \dot \Psi(t) , \Psi(t)\right\rangle = 
\frac{i}{\hbar}\left\langle H \Phi(t) , \Phi(t)\right\rangle + 
\frac{i}{\hbar} \left\langle  H \Phi(t) , G\underline\xi(t) \right\rangle. 
\end{equation}
Since  $-D_{m}G=\delta_{y}\otimes\II$ we have that 
\begin{equation}\label{where}
\frac{i}{\hbar} \left\langle  H \Phi(t) , G\underline\xi(t) \right\rangle = -  \frac{i}{\hbar} \left\langle   \Phi(y,t) ,\underline\xi(t) \right\rangle_{\CO^2} = - i   \left\langle  \q(t)  ,\left(\left(A(\bar\q,\q) +\frac{A(\bar\q,\q)\sigma_3A(\bar\q,\q)}{2c}\right)\q\right)(t) \right\rangle_{\CO^2},
\end{equation}
where we used Eq. \eqref{xiq}, the boundary condition in \eqref{dominio}, and the fact that $A$ is selfadjoint. Using the latter identity in Eq.  \eqref{boh}, and noticing that $\Im \left\langle H \Phi , \Phi\right\rangle = 0$ and $\Im\left\langle  \q  ,\left(A +\frac{A\sigma_3A}{2c}\right)\q \right\rangle_{\CO^2} = 0$, we conclude that $ \Re\left\langle \dot \Psi(t) , \Psi(t)\right\rangle =0$, which in turn implies that the $L^2$-norm is conserved. 
\end{proof}

To state the conservation of the energy we look at $H^{nl}_A$ as an Hamiltonian vector field with respect to the couple of canonical coordinates $\Psi$ and $\overline\Psi$. For this reason, in the next theorem we use the notation $A(\q) = \cA(\bar\q,\q)$. 
\begin{theorem}[Energy conservation]\label{t:energy} Assume that $\cA(\bar \q,\q) = \cA(\q,\bar \q)$,  and let $\Psi_\circ \in D(H_\cA^{nl})$. Then the energy  
\begin{equation}\label{E}
E(\Psi) = \left\langle\Psi,H_\cA^{nl} \Psi\right\rangle  -  \hbar \left\langle\q,\cA(\bar\q,\q)  \q \right\rangle_{\CO^2} + \hbar W(\bar{\underline q}, \underline q) ,
\end{equation}
where $W:\CO^4 \to \RE$ is such that $W(\bar\q,\q)=W(\q,\bar\q)$, and 
\begin{equation}\label{W}
\nabla_{\bar{\underline q}}W(\bar{\underline q} , \underline q) =  \cA(\bar\q,\q)\q   ,
\end{equation}
is conserved along the flow associated to the  Cauchy problem \eqref{cauchynonlin}. 
\end{theorem}
\begin{proof}
As first step we rewrite the energy functional in a different form. Recall that $\Psi \in D(H_\cA^{nl})$ can  be  decomposed as in Eq. \eqref{dec}. By Eq. \eqref{HAnl} it follows that  
\[\left\langle\Psi,H_\cA^{nl} \Psi\right\rangle = \left\langle\Phi,H \Phi\right\rangle + \left\langle G\underline \xi,H \Phi\right\rangle.
\] 
Repeating the calculations in Eq. \eqref{where} one obtains  
\[ \left\langle G\underline\xi  ,  H \Phi  \right\rangle = - \left\langle   \Phi(y) ,\underline\xi \right\rangle_{\CO^2} = - \hbar  \left\langle  \q  ,\left(\cA(\bar\q,\q) +\frac{\cA(\bar\q,\q)\sigma_3\cA(\bar\q,\q)}{2c}\right)\q \right\rangle_{\CO^2}. 
\]
Hence, for any state $\Psi\in D(H_\cA^{nl})$,  the energy functional can be written as 
\[E(\Psi) = \left\langle\Phi,H \Phi\right\rangle  -  2\hbar \left\langle\q,\left(\cA(\bar\q,\q) + \frac{\cA(\bar\q,\q)\sigma_{3}\cA(\bar\q,\q)}{4c}\right) \q \right\rangle_{\CO^2} +\hbar W(\bar{\underline q}, \underline q) .
\]
Next we compute the time derivative of the $E(\Psi(t))$  when  $\Psi(t)$ is  the solution of the Cauchy problem \eqref{cauchynonlin}. By using again the decomposition in Eq. \eqref{dec}, we have that
\[\begin{aligned}\frac{d}{dt}  & \langle\Phi(t),H \Phi(t)\rangle \\
&  = \lim_{s\to 0} \frac1s\left(   \langle\Phi(t+s),H \Phi(t+s)\rangle  -    \langle\Phi(t),H \Phi(t)\rangle \right) \\ 
 &  = \lim_{s\to 0}  \frac1s\left(      \langle\Phi(t+s),H \Phi(t+s)\rangle  -    \langle\Phi(t),H \Phi(t+s)\rangle +   \langle\Phi(t),H \Phi(t+s)\rangle  -    \langle\Phi(t),H \Phi(t)\rangle  \right) \\ 
 &= 2\Re \langle\dot \Phi(t),H \Phi(t)\rangle, \end{aligned}\]
where we used the fact that $\Phi(t)$ is in $D(H)$ for all $t\geq0$, is a  continuous  function of $t$,  and that $H$ is selfadjoint. 

 In what follows, to shorten the notation, we do not make explicit the dependence of functions on $t$. We have that 
\[
2\Re \langle\dot \Phi,H \Phi\rangle = 2\Re\left\langle\dot{\underline\xi},\Phi(y) \right\rangle_{\CO^2},  \]
where we used the fact that $\dot \Phi$ satisfies the equation in \eqref{Creg}, and that $-D_{m}G=\delta_{y}\otimes\II$. From the relations \eqref{xiq} and \eqref{Gmedxi} we have that 
\[\Phi(y) = \q  +\frac{\sigma_{3}\underline\xi}{2\hbar c} = \left(\II+\frac{\sigma_{3}\cA(\bar\q,\q)}{2c}\,\right )\q .\]
Hence, 
\[\begin{aligned}&\frac{d}{dt}   \left\langle\Phi,H \Phi\right\rangle \\ 
   =& 
2\hbar \Re\left\langle\frac{d}{dt}\left(\cA(\bar\q,\q)\q\right),   \left(\II+\frac{\sigma_{3}\cA(\bar\q,\q)}{2c}\,\right )\q\right\rangle_{\CO^2}  \\ 
= &
2\hbar \left\langle\q, \dot \cA(\bar\q,\q)\q\right\rangle_{\CO^2} + 2\hbar \Re\left\langle\dot\q, \cA(\bar\q,\q)\q\right\rangle_{\CO^2}+\frac{\hbar}{c} \Re\left\langle\q,\dot \cA(\bar\q,\q)\sigma_{3}\cA(\bar\q,\q)\q\right\rangle_{\CO^2}\\
& + \frac{\hbar}{c} \Re\left\langle\dot \q,   \cA(\bar\q,\q)\sigma_{3}\cA(\bar\q,\q)\q\right\rangle_{\CO^2}
 \end{aligned}\]
where we used the fact that $\cA(\bar\q,\q)$ and $\dot \cA(\bar\q,\q)$ are selfadjoint. We note that 
\[\begin{aligned}
\left\langle\q,\dot \cA(\bar\q,\q)\q\right\rangle_{\CO^2} = &\left\langle\q,\left(\dot{\bar\q}\cdot \nabla_{\bar\q}\cA(\bar\q,\q)\right)\q\right\rangle_{\CO^2}  + 
 \left\langle\q,\left(\dot{\q}\cdot \nabla_{\q}\cA(\bar\q,\q)\right)\q\right\rangle_{\CO^2} \\ 
 =& 2\Re\left\langle\q,\left(\dot{\bar\q}\cdot \nabla_{\bar\q}\cA(\bar\q,\q)\right)\q\right\rangle_{\CO^2},
\end{aligned}\]
where we used $\overline{\nabla_q \cA_{i,j}(\bar \q,\q)} = \nabla_{\bar \q }\cA_{j,i}(\bar \q,\q)$, which is a consequence of the assumption $\cA(\bar \q,\q) = \cA(\q,\bar \q)$ and of the fact that $\cA$ is selfadjoint. 
From the latter identity, it follows that  
\[
\left\langle\q,\dot \cA(\bar\q,\q)\q\right\rangle_{\CO^2} = 2\Re\left( \left\langle\dot\q,\nabla_{\bar\q}\langle\q,\cA(\bar\q,\q)\q\rangle_{\CO^2}\right\rangle_{\CO^2} - \left\langle\dot\q, \cA(\bar\q,\q)\q\right\rangle_{\CO^2} \right).
\]
In a similar way, by using the identity $\overline{\langle\q,(\nabla_{\q}\cA)\sigma_3 \cA q\rangle_{\CO^2}} = \langle\q,\cA\sigma_3(\nabla_{\bar\q}\cA)q\rangle_{\CO^2}$, one can prove that   
\[ \Re\left\langle\q,\dot \cA(\bar\q,\q)\sigma_{3}\cA(\bar\q,\q)\q\right\rangle_{\CO^2} = \Re\left(\left\langle\dot \q,\nabla_{\bar \q} \langle\q,\cA(\bar\q,\q)\sigma_{3}\cA(\bar\q,\q)\q\rangle_{\CO^2} \right\rangle_{\CO^2}  - \langle\dot \q,\cA(\bar\q,\q)\sigma_{3}\cA(\bar\q,\q)\q \rangle_{\CO^2}\right).
\]
Hence, 
\[ \frac{d}{dt}   \langle\Phi,H \Phi\rangle  = 4\hbar \Re \left\langle\dot \q, \nabla_{\bar\q}\left\langle\q,\left(\cA(\bar\q,\q)+ \frac{\cA(\bar\q,\q)\sigma_{3}\cA(\bar\q,\q)}{4c}\right) \q \right\rangle_{\CO^2}    \right\rangle_{\CO^2} - 2 \hbar \Re \left\langle\dot \q, \cA(\bar\q,\q)\q   \right\rangle_{\CO^2}.\] 
Taking into account the fact that 
\[ \begin{aligned}  & 2\hbar \frac{d}{dt} \left\langle\q,\left(\cA(\bar\q,\q) + \frac{\cA(\bar\q,\q)\sigma_{3}\cA(\bar\q,\q)}{4c}\right) \q \right\rangle_{\CO^2}  \\ 
 = &  4\hbar \Re \left\langle\dot \q, \nabla_{\bar\q}\left\langle\q,\left(\cA(\bar\q,\q)+ \frac{\cA(\bar\q,\q)\sigma_{3}\cA(\bar\q,\q)}{4c}\right) \q \right\rangle_{\CO^2}  \right\rangle_{\CO^2}
 ,
 \end{aligned}\]
and that 
\[\frac{d}{dt} W(\bar{\underline q}, \underline q) = 2 \Re \langle\dot\q, \nabla_{\bar q}W(\bar{\underline q}, \underline q)\rangle_{\CO^2} ,  \]together with Eqs. \eqref{E} and \eqref{W}, 
it follows that $\frac{d}{dt} E[\Psi] = 0$. 
\end{proof}

\section{Examples.} 
Given $M:\CO^{2}\to\CO^{2}$ self-adjoint, let $\phi(\z):=\langle\z,M\z\rangle_{\CO^{2}}$ be the corresponding quadratic form. We suppose that 
\[
A(\z)= A_{\circ}(\phi(\z))\,,
\] with 
$A_{\circ}$ such that  
\be\label{h1}
A_{\circ}(x)MA_{\circ}(x)=a(x)M\,,\quad a:\RE\to\RE\,.
\ee
Therefore   $\phi(A(\z)\z)=a(\phi(\z))\phi(\z)$ and 
\[
\phi(F_{A}(\z))=\phi(\z)+\frac1{4c^{2}}\,\phi(A(\z)\z)=f_{a}(\phi(\z))\,,\qquad f_{a}(x):=\left(1+\frac{1}{4c^{2}}\,a(x)\right)x\,.
\]
So, 
if $x\mapsto a(x)x$ is $C^{1}$ and
\be\label{h2}
\lim_{|x|\to+\infty}\left|1+\frac{1}{4c^{2}}\,a(x)\right|\,|x|=+\infty\,,
\ee
\be\label{h3}
1+\frac{1}{4c^{2}}\,\frac{d\ }{dx}\,(a(x)x)\not=0\,,
\ee
then $f_{a}:\RE\to\RE$ is a $C^{1}$-diffeomorphism by Hadamard's theorem and $F_{A}$ has a global inverse given by  
\[
F_{A}^{-1}(\z)=\left(\II+\frac{i}{2c}\,A_{\circ}\big(f^{-1}_{a}(\phi(\z))\big)\right)^{-1}\z\,.
\]
Therefore $F_{A}^{-1}$ is a $C^{1}$-diffeomorphism (and hence Assumption \ref{a:F} holds) whenever $x\mapsto A_{\circ}(x)$ is a $C^{1}$ map such that \eqref{h2} and \eqref{h3} hold. Let us notice that, as the next example shows, Assumption \ref{a:F} can hold true under weaker conditions. 
\subsection{Nonlinear Gesztesy-\v Seba models.}
The two simplest models are the ones in which $M=M_{\pm}:=\frac12(\II\pm\sigma_{3})$, i.e. either $\phi(\z)=\phi_{+}(\z):=|z_{1}|^{2}$ or $\phi(\z)=\phi_{-}(\z):=|z_{2}|^{2}$. These give nonlinear versions of the models introduced in \cite{Gesztesy:1987gn}. One has that \eqref{h1} holds if and only if 
\[
A_{\circ}(x)=\alpha(x)\,M_{\pm}\,,\qquad\alpha:\RE\to\RE\,,\qquad a=\alpha^{2}\,.
\]
By straightforward computation one gets that the Jacobian determinant of $F_{A}$ never vanishes if and only if 
$$
1+\frac1{4c^{2}}\,\frac{d\ }{dx}(\alpha^{2}(x)x)\not=0\,.
$$ 
So, for example,  Assumption \ref{a:F} holds true whenever $\alpha(x)=\kappa\, x^{2\sigma}$, $\kappa\in\RE$ and $\sigma\in\RE_{+}$; the case $\sigma\in (0,1/2)$ shows that 
Assumption \ref{a:F} can be true even if $x\mapsto A_{\circ}(x)$ is not a $C^{1}$ map.\par 
Let us notice that, in the linear case, the 
nonrelativistic limit of the ``$+$'' case gives a Schr\"odinger operator with a delta interaction of strength $\alpha$, whereas  the 
nonrelativistic limit of the ``$-$'' case gives a Schr\"odinger operator with a delta prime  interaction of strength $-1/\alpha$ (see \cite{Benvegnu:1994ci}).
\subsection{Example.} If $M=\II$, i.e. $\phi(\z)=|\z|^{2}$, then \eqref{h1} holds if and only if either
\[
A_{\circ}(x)=\alpha(x)\,\II\,,\qquad \alpha:\RE\to\RE\,,\qquad a=\alpha^{2}\,,
\]
or
\[
A_{\circ}(x)=\begin{bmatrix}\alpha(x)&\ \,\,\gamma(x)\\
\bar\gamma(x)&-\alpha(x)\end{bmatrix},\quad \alpha:\RE\to\RE\,,\quad \gamma:\RE\to\CO\,,\quad a=\alpha^{2}+|\gamma|^{2}\,.
\]
As a special case,  taking $A(\z)=3\beta |\z|^{2}\II$, $\beta\in\RE$, i.e. $W(\z)=\frac32\,\beta\,|\z|^{4}$, one obtains  the ``concentrated'' version of the ``Brag Resonance model, see \cite{Pelinovsky:2010vh}. There a different representation of the Dirac operator is used; with reference to Remark \ref{equiv}, it corresponds to the choice 
\[
U=\frac{1}{\sqrt{2}}
\left[\begin{matrix}1 & -1 \\1 & \ 1\end{matrix}\right]\,,
\] 
so that 
\[
\tilde D_{m}=-i\hbar c\, \sigma_{3}\frac{d}{dx}-mc^{2}\sigma_{1}\,.
\]  
In such a representation the corresponding potential is given by
\[
\tilde W(\z)=\beta \big(|\z|^{4}+2\,|z_{1}|^{2}|z_{2}|^{2}\big)\,.
\]
 
\subsection{Example.} If $M=\sigma_{1}$, i.e. $\phi(\z)=z_{1}\bar z_{2}+\bar z_{1}z_{2}$, then \eqref{h1} holds if and only if either
\[
A_{\circ}(x)=\gamma(x)\,\sigma_{1}\,,\quad \gamma:\RE\to\RE\,,\qquad a=\gamma^{2}\,,
\]
or
\[
A_{\circ}(x)=\begin{bmatrix}\alpha_{1}(x)&\gamma(x)\\
\bar\gamma(x)&\alpha_{2}(x)\end{bmatrix},\quad \alpha_{1}:\RE\to\RE\,,\ \alpha_{2}:\RE\to\RE\,,\ \gamma:\RE\to i\RE\,,\quad a=\alpha_{1}\alpha_{2}-|\gamma|^{2}\,.
\]
\subsection{Example.} If $M=\sigma_{2}$, i.e. $\phi(\z)=i(z_{1}\bar z_{2}-\bar z_{1}z_{2})$, then \eqref{h1} holds if and only if either 
\[
A_{\circ}(x)=\gamma(x)\,\sigma_{2}\,,\qquad \gamma:\RE\to\RE\,,\qquad a=\gamma^{2}\,,
\]
or
\[
A_{\circ}(x)=\begin{bmatrix}\alpha_{1}(x)&\gamma(x)\\
\gamma(x)&\alpha_{2}(x)\end{bmatrix}\,,\quad \alpha_{1}:\RE\to\RE\,,\ \alpha_{2}:\RE\to\RE\,,\ \gamma:\RE\to\RE\,,\quad
a=\alpha_{1}\alpha_{2}-\gamma^{2}\,.
\]
\subsection{Soler-type Models.} If $M=\sigma_{3}$, i.e. $\phi(\z)=|z_{1}|^{2}-|z_{2}|^{2}$, then \eqref{h1} holds if and only if either
\[
A_{\circ}(x)=\alpha(x)\,\sigma_{3}\,,\qquad \alpha:\RE\to\RE\,,\qquad
a=\alpha^{2}\,,
\]
or
\[
A_{\circ}(x)=\begin{bmatrix}\alpha(x)&\gamma(x)\\
\bar\gamma(x)&\alpha(x)\end{bmatrix}\,,\quad \alpha:\RE\to\RE\,,\quad \gamma:\RE\to\CO\,,\quad a=\alpha^{2}-|\gamma|^{2}\,.
\]
As a special case,  taking $A(\z)=4\left( |z_{1}|^{2}-|z_{2}|^{2}\right)\sigma_{3}$, i.e $W(\z)=2\left( |z_{1}|^{2}-|z_{2}|^{2}\right)^{2}$, one obtains  the ``concentrated'' version of the ``massive Gross-Neveu model'' (see \cite{Gross:1974gs, Pelinovsky:2010vh}) which corresponds to the 1-D Soler Model (see \cite{Soler:1970}). Notice that, with respect to the representation $\tilde D_{m}$ of the Dirac operator used in \cite{Pelinovsky:2010vh}, the potential is given by 
$\tilde W(\z)=2(\overline z_{1}z_{2}+z_{1}\overline z_{2})^{2}$.  
\appendix

\section{A regularity result for the free evolution. \label{app:A}}
Here we recall the basic results for the Cauchy problem for the 1-D free Dirac operator 
\be\label{cauchyfree}
\begin{cases}
i\hbar\frac{d}{d t}\Psi^{f}(t)=D_{m}\Psi^{f}(t)&\\
\Psi^{f}(0)=\Psi_{\circ}\,.&
\end{cases}
\ee
By 
\[
\left(-\frac{i}{\hbar}\,D_{m}\right)^{2}=K_{m}:=c^{2}\frac{d^{2}}{dx^{2}}-\frac{m^{2}c^{4}}{\hbar^{2}}\,,
\]
such a Cauchy problem is equivalent to 
\[
\begin{cases}
\frac{d^{2}}{d t^{2}}\Psi^{f}(t)=K_{m}\Psi^{f}(t)&\\
\Psi^{f}(0)=\Psi_{\circ}&\\
\frac{d}{d t}\Psi^{f}(0)=-\frac{i}{\hbar}D_{m}\Psi_{\circ}\,.&
\end{cases}
\]
The solution of the Cauchy problem for the Klein-Gordon equation is known (see e.g. \cite[Section II.5.4]{Tikhonov:1963ti},  \cite[Section 4.1.3-3]{Polyanin:2015tx})
\[
\begin{split} 
&\Psi^{f}(x,t)=\frac{1}{2}\,(\Psi_{\circ}(x-ct)+\Psi_\circ(x+ct))\\
&-\frac{mc^{2}t}{2\hbar}\int_{x-ct}^{x+ct}d\xi\ \frac{J_1\left(\frac{mc}{\hbar}\,\sqrt{(ct)^2-(x-\xi)^2}\,\right)}{\sqrt{(ct)^2-(x-\xi)^2}}\ 
\Psi_\circ(\xi)\\
&+\frac{1}{2c}\int_{x-ct}^{x+ct}d\xi\,J_0\left(\frac{mc}{\hbar}\,\sqrt{(ct)^2-(x-\xi)^2}\,\right)\left(-\frac{i}{\hbar}D_{m}\Psi_{\circ}\right)
\end{split}
\]
Posing
\[
\Psi^{f}(t)=\begin {pmatrix}\psi_1^f(t)\\\psi_2^f(t)\end {pmatrix}\,,\quad
\Psi_{\circ}=\begin {pmatrix}\psi_1^\circ\\\psi_2^\circ\end {pmatrix}\,,
\]
integrating by parts and by $\frac{d\,}{dx}J_{0}(x)=-J_{1}(x)$, the solution can be rewritten in an equivalent way as
\be\label{free}
\begin{split}
&\psi_k^f(x,t)=\frac{1}{2}\left((\psi_k^\circ+\psi_j^\circ)(x-ct)+(\psi_k^\circ-\psi_j^\circ)(x+ct)\right)\\
&-\frac{m\,c}{2\,\hbar}\int_{x-ct}^{x+ct}\left(ct\ \frac{J_1\left(\frac{m\,c}{\hbar}\sqrt{(ct)^2-(x-\xi)^2}\,\right)}{\sqrt{(ct)^2-(x-\xi)^2}}-i(-1)^{k}J_0\left(\frac{m\,c}{\hbar}\sqrt{(ct)^2-(x-\xi)^2}\,\right)\right)\psi_k^\circ(\xi)\,d\xi\\
&-\frac{mc}{2\,\hbar}\int_{x-ct}^{x+ct}\frac{J_1\left(\frac{m\,c}{\hbar}\sqrt{(ct)^2-(x-\xi)^2}\,\right)}{\sqrt{(ct)^2-(x-\xi)^2}}\ (x-\xi)\psi_j^\circ(\xi)\,d\xi,
\qquad j,k =1,2\,;\quad j\neq k. 
\end{split}
\ee
Therefore, defining the matrix-valued kernel function 
\[
K(x,t) =- \frac{mc}{2\hbar} \left( i\sigma_{3}J_0\left(\frac{m\,c}{\hbar}\sqrt{(ct)^2-x^2}\,\right)+
(ct\II+x\sigma_{1})\ \frac{J_1\left(\frac{m\,c}{\hbar}\sqrt{(ct)^2-x^2}\,\right)}{\sqrt{(ct)^2-x^2}}\right)
\]
the solution of the Cauchy problem \eqref{cauchyfree} can be written as 
\begin{align}
\Psi^f(x,t)  = & \left(e^{-\frac i\hbar tH} \Psi_\circ\right)(x)   \nonumber \\ 
= & \frac12\left((\mathbb{1}+\sigma_1)\Psi_\circ(x-ct) + (\mathbb{1}-\sigma_1) \Psi_\circ(x+ct) \right)
+ \int_{x-ct}^{x+ct}  d\xi\,K(x-\xi,t)  \Psi_\circ(\xi) \,.\label{Ufree}
\end{align}

In the following proposition we establish the regularity properties of the map $t\mapsto\Psi^f(y,t)$.
\begin{proposition}\label{H1}
For any $\Psi_\circ\in H^{1}(\RE\backslash\{y\})\otimes\CO$ and  $T>0$, $\Psi^f(y,\cdot) \in H^1(0,T)$.
\end{proposition}
\begin{proof}We use identity \eqref{free},  which we rewrite as 
\begin{equation*}
\psi_k^f(y,t)= u_{1,k}(t) +u_{2,k}(t) +u_{3,k}(t) +u_{4,k}(t)
\end{equation*}
with 
\begin{equation*}
u_{1,k}(t) = \frac{1}{2}\left((\psi_k^\circ+\psi_j^\circ)(y-ct)+(\psi_k^\circ-\psi_j^\circ)(y+ct)\right) ,
\end{equation*}
\begin{equation*}
u_{2,k}(t) = -\frac{m\,c}{2\,\hbar}\int_{y-ct}^{y+ct} ct\ \frac{J_1\left(\frac{m\,c}{\hbar}\sqrt{(ct)^2-(y-\xi)^2}\,\right)}{\sqrt{(ct)^2-(y-\xi)^2}}\psi_k^\circ(\xi)\,d\xi ,
\end{equation*}
\begin{equation*}
u_{3,k}(t) =i(-1)^{k}\frac{m\,c}{2\,\hbar}\int_{y-ct}^{y+ct}J_0\left(\frac{m\,c}{\hbar}\sqrt{(ct)^2-(y-\xi)^2}\,\right)\psi_k^\circ(\xi)\,d\xi , 
\end{equation*}
\begin{equation*}
u_{4,k}(t) = -\frac{mc}{2\,\hbar}\int_{y-ct}^{y+ct}\frac{J_1\left(\frac{m\,c}{\hbar}\sqrt{(ct)^2-(y-\xi)^2}\,\right)}{\sqrt{(ct)^2-(y-\xi)^2}}\ (y-\xi)\psi_j^\circ(\xi)\,d\xi,
\end{equation*}
with $k,j=1,2$ and  $k\neq j$. 

We start by noting that, for $k=1,2$,
\begin{equation*}
\int_0^T |{\psi_k^{\circ}}'(y+ct)|^2 \,dt  = c \int_y^{y+cT}  |{\psi_k^{\circ}}'(s)|^2 \,ds \leq c \|{\psi_k^{\circ}}' \|_{L^2(y,+\infty)}^2.
\end{equation*}
Similarly 
\begin{equation*}
\int_0^T |{\psi_k^{\circ}}'(y-ct)|^2  \,dt  \leq c \|{\psi_k^{\circ}}' \|_{L^2(-\infty,y)}^2,
\end{equation*}
and for  $\int_0^T |{\psi_k^{\circ}}(y\pm ct)|^2  \,dt$. Hence,    
\begin{equation*}
\|u_{k,1} \|_{H^1(0,T)} \leq C \sum_{j=1}^2 \|  {\psi_j^{\circ}}\|_{H^1(\RE\backslash\{y\})} .  
\end{equation*}

Next we analyze the integral terms in \eqref{free}. Recall that 
\begin{equation}\label{bounds}
\|\psi^\circ_k\|_{L^\infty(y,+\infty)}^2 \leq 2  \|{\psi_k^{\circ}}\|_{L^2(y,+\infty)} \|{\psi_k^{\circ}}' \|_{L^2(y,+\infty)} \;,\qquad
\|\psi^\circ_k\|_{L^\infty(-\infty,y)}^2 \leq  2 \|{\psi_k^{\circ}} \|_{L^2(-\infty,y)}\|{\psi_k^{\circ}}' \|_{L^2(-\infty,y)}.
\end{equation}
We shall prove that, for $l=2,3,4$, $u_{l,k}'(t)$ is bounded for all $t \in [0,T]$, this in turn implies that $u_{l,k} \in H^1(0,T)$. 

We start with $u_{2,k}$. We split the integral on the intervals $(y-ct,y)$  and $(y,y+ct)$  and consider first  the integration for  $\xi\in (y,y+ct)$, we have that  
\begin{equation}\label{integral1}
\int_{y}^{y+ct} ct\ \frac{J_1\left(\frac{m\,c}{\hbar}\sqrt{(ct)^2-(y-\xi)^2}\,\right)}{\sqrt{(ct)^2-(y-\xi)^2}} \psi_k^{\circ}(\xi) \, d\xi = 
ct \int_{0}^{ct} \ \frac{J_1\left(\frac{m\,c}{\hbar}\sqrt{(2ct-\eta)\eta}\,\right)}{\sqrt{(2ct-\eta)\eta}} \psi_k^{\circ}(ct+y-\eta) \, d\eta .
\end{equation}
Taking the derivative with respect to $t$  we obtain 
\begin{equation}\label{derivative1}\begin{aligned}
&\frac{d}{dt}\int_{y}^{y+ct} ct\ \frac{J_1\left(\frac{m\,c}{\hbar}\sqrt{(ct)^2-(y-\xi)^2}\,\right)}{\sqrt{(ct)^2-(y-\xi)^2}} \psi_k^{\circ}(\xi) \, d\xi\\ = &
c\int_{0}^{ct} \ \frac{J_1\left(\frac{m\,c}{\hbar}\sqrt{(2ct-\eta)\eta}\,\right)}{\sqrt{(2ct-\eta)\eta}} \psi_k^{\circ}(ct+y-\eta) \, d\eta  \\ 
&+c^2t \frac{J_1\left(\frac{m\,c}{\hbar}\sqrt{(ct)^2}\,\right)}{\sqrt{(ct)^2}} \psi_k^{\circ}(y^+)  \\
&+ct \int_{0}^{ct} \ \frac{d}{dt} \left(\frac{J_1\left(\frac{m\,c}{\hbar}\sqrt{(2ct-\eta)\eta}\,\right)}{\sqrt{(2ct-\eta)\eta}} \right) \psi_k^{\circ}(ct+y-\eta) \, d\eta \\ 
&+c^2t \int_{0}^{ct} \ \frac{J_1\left(\frac{m\,c}{\hbar}\sqrt{(2ct-\eta)\eta}\,\right)}{\sqrt{(2ct-\eta)\eta}} {\psi_k^{\circ}}'(ct+y-\eta) \, d\eta. 
\end{aligned}
\end{equation}
For the first, second and fourth term at the r.h.s. we use the bounds \eqref{bounds} and the fact that for any $a\geq 0$, there exists a constant $C$ such that $|J_1(\sqrt{a})/\sqrt a|\leq C$. So that, for all $t\in [0,T]$, each of those terms is bounded by $C_T \|{\psi_k^{\circ}}\|_{H^1(y,\infty)} $, where $C_T$ is a constant which depends on $T$. 

For the third term at the r.h.s. of Eq. \eqref{derivative1} we use the fact that  for any $a,b\geq 0$, there exists a constant $C$ such that $|\frac{d}{da}J_1(\sqrt{ab})/\sqrt{ab}|\leq C/a$. We have that, for all $\eta \in[0,ct]$, 
\begin{equation*}
t\left| \frac{d}{dt} \frac{J_1\left(\frac{m\,c}{\hbar}\sqrt{(2ct-\eta)\eta}\,\right)}{\sqrt{(2ct-\eta)\eta}} \right| 
\leq \frac{C t }{2ct-\eta}  \leq C .  
\end{equation*}
Hence,  the third term at the r.h.s. of Eq. \eqref{derivative1} is also bounded by  $C_T \|{\psi_k^{\circ}}\|_{H^1(y,\infty)} $. The integral of the form \eqref{integral1} on the interval $(y-ct,y)$ is bounded in a similar way and we omit the details. We have proved that,  for all $t\in[0,T]$, $|u_{2,k}(t)| \leq C_T  \|{\psi_k^{\circ}}\|_{H^1(\RE\backslash\{y\})} $. 

The analysis of $u_{3,k}$ is straightforward. Splitting again the integration interval and taking the derivative, we obtain 
\begin{equation*}\begin{aligned}
& \frac{d}{dt}\int_{y}^{y+ct}J_0\left(\frac{m\,c}{\hbar}\sqrt{(ct)^2-(y-\xi)^2}\,\right)\psi_k^\circ(\xi)\,d\xi \\ 
= & c \psi_k^\circ(y+ct) -  \frac{m\,c}{\hbar} c^2 t \int_{y}^{y+ct}   \frac{J_1\left(\frac{m\,c}{\hbar}\sqrt{(ct)^2-(y-\xi)^2}\,\right)}{\sqrt{(ct)^2-(y-\xi)^2}} \psi_k^\circ(\xi)\,d\xi.  
\end{aligned}\end{equation*}
And similarly for the integral on the interval $(y-ct,y)$.  By inequalities \eqref{bounds}, and since $J_1(a)/a$ is  bounded for all $a\geq 0$, we conclude  that  for all $t\in[0,T]$, $|u_{3,k}(t)| \leq C_T  \|{\psi_k^{\circ}}\|_{H^1(-\infty,y) \oplus  H^1(y,\infty)}$. 

We are left to analyze $u_{4,k}$. Also in this case we split the integral on the intervals $(y-ct,y)$  and $(y,y+ct)$, and take the derivative. We obtain 
\begin{equation}\label{maybe}\begin{aligned}
& \frac{d}{dt}\int_{y}^{y+ct}\frac{J_1\left(\frac{m\,c}{\hbar}\sqrt{(ct)^2-(y-\xi)^2}\,\right)}{\sqrt{(ct)^2-(y-\xi)^2}}\ (y-\xi)\psi_j^\circ(\xi)\,d\xi  \\ 
 = &   
- \frac{m\,c^2}{2\hbar}  (ct )\psi_j^\circ(y+ct)   +  
c^2 t \int_{y}^{y+ct}\frac{d}{d\xi}\left(\frac{J_1\left(\frac{m\,c}{\hbar}\sqrt{(ct)^2-(y-\xi)^2}\,\right)}{\sqrt{(ct)^2-(y-\xi)^2}} \right) \psi_j^\circ(\xi)\,d\xi  ,
\end{aligned} \end{equation}
where we used the identity 
\begin{equation*}
(y-\xi) \frac{d}{dt}\left(\frac{J_1\left(\frac{m\,c}{\hbar}\sqrt{(ct)^2-(y-\xi)^2}\,\right)}{\sqrt{(ct)^2-(y-\xi)^2}} \right)  = c^2 t \frac{d}{d\xi}\left(\frac{J_1\left(\frac{m\,c}{\hbar}\sqrt{(ct)^2-(y-\xi)^2}\,\right)}{\sqrt{(ct)^2-(y-\xi)^2}} \right). 
\end{equation*}
In equation \eqref{maybe} we  integrate by parts, and obtain 
\begin{equation*}\begin{aligned}
 &\frac{d}{dt}\int_{y}^{y+ct}\frac{J_1\left(\frac{m\,c}{\hbar}\sqrt{(ct)^2-(y-\xi)^2}\,\right)}{\sqrt{(ct)^2-(y-\xi)^2}}\ (y-\xi)\psi_j^\circ(\xi)\,d\xi   \\ 
 =  &
- c J_1\left(\frac{m\,c^2}{\hbar} t \right) \ \psi_j^\circ(y^+)   - 
c^2 t \int_{y}^{y+ct} \frac{J_1\left(\frac{m\,c}{\hbar}\sqrt{(ct)^2-(y-\xi)^2}\,\right)}{\sqrt{(ct)^2-(y-\xi)^2}}  \ {\psi_j^\circ}'(\xi)\,d\xi .
\end{aligned} \end{equation*}
A similar result holds true for the integral on the interval $(y-ct,y)$. By using again the bounds \eqref{bounds} and the boundedness of  $J_1(a)/a$ we obtain  $|u_{3,k}(t)| \leq C_T  \|{\psi_k^{\circ}}\|_{H^1(\RE\backslash\{y\})}$, and  this concludes the proof of the proposition. 
\end{proof}

\end{document}